\documentclass[aps,prx,10pt,nofootinbib]{revtex4-2}

\usepackage{graphicx}
\usepackage{dcolumn}
\usepackage{bm}
\usepackage{mathtools}
\usepackage{amsmath}
\usepackage{amsthm}
\usepackage{amsfonts}
\usepackage{amssymb}
\usepackage{dsfont}
\usepackage{enumitem} 
\setlist{nolistsep}
\usepackage{booktabs}
\usepackage{multirow}
\newcommand{\citationneeded}[1][]{\textsuperscript{\color{blue} [citation needed]}}

\usepackage[table,usenames,dvipsnames]{xcolor}
\usepackage[colorlinks=true,
			hyperindex, 
		    linkcolor = NavyBlue,
		    anchorcolor = hyperrefblue,
		    citecolor = Green,
		    filecolor = hyperrefblue,
		    urlcolor = NavyBlue,
			breaklinks=true]{hyperref}

\usepackage{braket}
\newtheorem{theorem}{Theorem}

\newcommand{\e}{\ensuremath\mathrm{e}}
\renewcommand{\i}{\ensuremath\mathrm{i}}

\DeclareMathOperator{\Herm}{Herm} 

\newcommand{\1}{\mathds{1}}

\newcommand{\mc}[1]{\mathcal{#1}}

\newcommand{\mcH}{\mc{H}}
\renewcommand{\H}{\mcH}

\renewcommand{\vec}[1]{\boldsymbol{#1}}

\newcommand{\myleft}{\mathopen{}\mathclose\bgroup\left}
\newcommand{\myright}{\aftergroup\egroup\right}

\newcommand{\iiiNorm}[1]{{\left\vert\kern-0.25ex\left\vert\kern-0.25ex\left\vert #1 
    \right\vert\kern-0.25ex\right\vert\kern-0.25ex\right\vert}}
\newcommand{\sandwich}[3]
  {\left\langle  #1 \right| #2 \left| #3 \right\rangle}

\newcommand{\class}[1]{{\ensuremath{\mathsf{#1}}}}
\newcommand{\NP}{\class{NP}}

\newcommand{\MaxCut}{\class{MaxCut}}
\newcommand{\ThreeSAT}{\class{3-SAT}}

\newtheorem{prob}[theorem]{Problem}
\newtheorem{prop}[theorem]{Proposition}
\newtheorem{coro}[theorem]{Corollary}
\newtheorem{conj}[theorem]{Conjecture}
\renewcommand{\P}{\class{P}}

\newcommand{\OEx}[1]{\left<O(#1)\right>}

\newcommand{\CA}{\mathcal{A}}

\newcommand{\CU}{\mathcal{U}}

\begin{document}


\title{Undecidable problems associated with variational quantum algorithms}
\author{Georgios Korpas}

\affiliation{Department of Computer Science, Czech Technical University in Prague, Czechia}%
\affiliation{Quantum Technologies Group, Innovation \& Ventures, HSBC, Singapore}
\affiliation{Archimedes Research Unit on AI, Data Science and Algorithms, Athena RC, Marousi, Greece}%

\author{Vyacheslav Kungurtsev}
\author{Jakub Mareček}
\affiliation{Department of Computer Science, Czech Technical University in Prague, Czechia}%

\date{\today}


\begin{abstract}
Variational Quantum Algorithms (VQAs), such as the Variational Quantum Eigensolver (VQE) and the Quantum Approximate Optimization Algorithm (QAOA), are widely studied as candidates for near-term quantum advantage. Recent work has shown that training VQAs is \NP-hard in general. In this paper, we present a conditional result suggesting that the training of VQAs is undecidable, even in idealized, noiseless settings. We reduce the decision version of the digitized VQA training problem—where circuit parameters are drawn from a discrete set—to the question of whether a universal Diophantine equation (UDE) has a root. This reduction relies on encoding the UDE into the structure of a variational quantum circuit via the matrix exponentials. The central step involves establishing a correspondence between the objective function of the VQA and a known UDE of 58 variables and degree 4. Our main result is conditional on a natural conjecture: that a certain system of structured complex polynomial equations—arising from the inner product of a VQA circuit output and a fixed observable—has at least one solution. We argue this conjecture is plausible based on dimension-counting arguments (degrees of freedom in the Hamiltonians, state vector, and observable), and the generic solvability of such systems in algebraic geometry over the complex numbers. Under this assumption, we suggest that deciding whether a digitized VQA achieves a given energy threshold is undecidable. This links the limitations of variational quantum algorithms to foundational questions in mathematics and logic, extending the known landscape of quantum computational hardness to include uncomputability. Additionally, we establish an unconditional undecidability result for VQA convergence in open quantum systems.
\end{abstract}

\maketitle


\section{Introduction}

Quantum algorithms, especially those implementable on near-term hardware, have garnered significant attention in recent years~\cite{preskill2018quantum, cezero2021, huang2023near}. Many financial institutions are investing in quantum computing, motivated by its potential applications in risk assessment, portfolio optimization, and stochastic control \cite{egger2020quantum,  gilliam2021grover, herman2022survey, intallura2023survey} while many other industries follow this path as well. Variational Quantum Algorithms (VQAs), encompassing the Variational Quantum Eigensolver (VQE) \cite{peruzzo2014variational}, the Quantum Approximate Optimization Algorithm (QAOA) \cite{farhi2014quantum}, and their respective variations, have sometimes been advocated as a methodology for addressing \NP-hard optimization problems. See \cite{cezero2021} for an extensive review. These algorithms operate by iteratively evaluating a parameterized quantum circuit and updating its parameters to minimize a cost function—typically, the expectation value of an observable (see App. \ref{sec:appVQA}).

However, the expectation value of the quantum observable, as measured by a quantum computer, depends on specific parameters provided to a classical optimizer. The objective of this optimizer is to explore the parameter landscape, adjusting the parameters to minimize the observable’s expectation value. These refined parameters are then anticipated to yield a reduced estimate of the expectation value of the quantum observable, ultimately converging to the genuine minimum, ideally. This process—commonly referred to as VQA training—forms a hybrid quantum-classical loop.

The classical computation of parameters for the variational form is often achieved through the use of gradient-based approaches \cite{schuld2019evaluating} or quasi-Newton methods \cite{byrd1995limited}. In particular, it has been observed that even second-order techniques often converge to local minima and are not guaranteed to reach global optima \cite{cerezo2020impact}, and the computation of the optima, even to a predetermined precision, is \NP-hard \cite{bittel2021training}.
Independently, there are lower bounds \cite{Bravyi2021,bravyi2020obstacles} on the depth of parametrized circuits that one needs to consider when analyzing the efficiency of variational quantum algorithms. The impact of noise on the iteration complexity of VQAs was studied in \cite{kungurtsev2022iteration}.

In this work, we show that training Variational Quantum Algorithms (VQAs)—under certain assumptions—can be undecidable, extending the known complexity barriers beyond \NP-hardness. Specifically, we consider the conditions under which the classical computation of the parameters of the variational form becomes uncomputable. We also examine related problems in training VQAs and analyze whether QAOA converges in the presence of noise.

Our main result, Theorem~\ref{thm:digivqeundecidable}, connects the decision version of training VQAs to Hilbert's 10th problem by reducing it to the solvability of a universal Diophantine equation, assuming the existence of a common solution for a certain system of complex polynomial equations. In Conjecture~\ref{conj:system}, we argue that such a solution likely exists, based on degrees-of-freedom heuristics and the typical behavior of structured polynomial systems. We then show that this result extends naturally to QAOA.

Our investigation focuses primarily on worst-case behavior. These findings do not rule out the possibility that specific subclasses of VQAs may still be successfully trainable under more favorable conditions, which we briefly discuss in Section~\ref{sec:conclusions}.

\vspace{1em}


\paragraph*{Notation.} 
We use the notation $[n]\coloneqq \{1,\dots,n\}$. 
The Pauli matrices are denoted by $\sigma_x$, $\sigma_y$, and $\sigma_z$ respectively. 
By $\Herm(\mathbb{C}^n)$ we denote the space of Hermitian $n \times n$ matrices, i.e., matrices $A \in \mathbb{C}^{n\times n}$ that satisfy $A= A^\dagger$, where $\dagger$ denotes complex conjugation. A complex
Hilbert space $\mathcal{H}$ of dimension $n$ is isomorphic to $\mathbb{C}^n$. By
$\|A\|$ we refer to the operator norm of an operator $A \in {\rm End}(\mathcal{H}) \cong \Herm(\mathbb{C}^n)$ while by $U(n) \subset {\rm GL}_n(\mathbb{C})$ we denote the degree $n$ unitary group. For a field $k$, we denote by $k[\vec x]$ the ring of polynomials with coefficients in $k$ where $\vec x = (x_1,\ldots, x_m)$, for some integer $m>0$, .


\section{Correspondence of VQA to Polynomial Equation Systems}

Initially, let us focus on the VQA minimization problem; specifically we state the VQA minimization problem as termed in \cite{bittel2021training} and also a digitized VQA minimization problem. These problems will be fundamental in the subsequent analysis.

\begin{prob}[{VQA} minimization problem, paraphrasing \cite{bittel2021training}]\label{p:VQA}\hfill
  \begin{description}
     \item[Instance] An initial state $\ket{\Psi_0}\in \mathbb{C}^n$, a set of generators $\{H_i \}_{i \in \{1,\dots, L\}}\in \Herm(\mathbb{C}^n)$, where $L \in \mathbb{Z}_{>0}$ is the number of layers and an observable $O\in\Herm(\mathbb{C}^n)$. 
     \item[Optimization version] For $\ket{\Psi(\vec \phi)}\coloneqq U_L(\phi_L)\cdots U_1(\phi_1)\ket{\Psi_0}$ with $U_i(\phi_i)=\e^{-\i H_i \phi_i}$, 
    find $\vec \phi \in \mathbb{R}^L$ that minimizes $\OEx{\vec\phi}\coloneqq~\sandwich{\Psi(\vec \phi)}{O}{\Psi(\vec \phi)}$. 
    \item[Decision version for $a \in \mathbb{R}$] For $\ket{\Psi({\vec \phi})}\coloneqq U_L(\phi_L)\cdots U_1(\phi_1)\ket{\Psi_0}$ with $U_i(\phi_i)~=~\e^{-\i H_i \phi_i}$, determine if there exists $\vec \phi \in \mathbb{R}^L$ for which $\sandwich{\Psi(\vec \phi)}{O}{\Psi(\vec \phi)} \leq a$.
\end{description}
\end{prob}

\paragraph*{Example.} For illustration purposes, consider the $L=1$ case of Problem~\ref{p:VQA}. There, one seeks a parameter $\phi\in\mathbb{R}$ such that the state given by
$\ket{\Psi({\phi})}\coloneqq  U(\phi)\ket{\Psi_0}\coloneqq\e^{-\i H \phi}\ket{\Psi_0}$ satisfies
$\sandwich{\Psi({\phi})}{O}{\Psi({\phi})}\leq a \in \mathbb{R}$. The operator $O$ can be decomposed as 
\begin{align}\label{eq:O}
    O=\sum\limits_{j=1}^n \lambda_j P_{\lambda_j},
\end{align}
with $P_{\lambda_j}$ the associated eigenspace projection operators with eigenvalue $\lambda_j$, the associated unconstrained optimization problem becomes: 
\begin{equation}\label{eq:classicalopt1}
    \min\limits_{\phi\in\mathbb{R}}\, \sum\limits_{j=1}^n \lambda_j \left\| P_{\lambda_j} \e^{-\i H \phi}\ket{\Psi_0}\right\|^2
\end{equation}

\noindent
which is a continuously differentiable function of a real valued variable $\phi$.
Let $n=2$, $\ket{\Psi_0}=\alpha \ket{0}+\beta \ket{1}$, $O=\begin{psmallmatrix}2&0\\0& -1\end{psmallmatrix}$ and $H=\sigma_x$, so $\cos(H\phi) = \cos(\phi) I$ and $\sin(H\phi)=\sin(\phi)\sigma_x$.  Then, Problem~\eqref{eq:classicalopt1} becomes,
\[
\min\limits_{\phi\in\mathbb{R}} (2\alpha^2-\beta^2)\cos^2\phi -(2\beta^2-\alpha^2)\sin^2\phi
\]
(which fits the framework of Theorem~\ref{th:optund}) and does not have any known exploitable structure such as convexity. This already hints at the undecidability of VQA.

Nevertheless, we know the form of the optimization problem, and although one can consider that it is unlikely that a special structure-exploiting algorithm for minimization could be found, this argument cannot guarantee that one does not exist.

Thus, to prove the undecidability, let us consider a floating-point representation of the parametrization of the variational form.
Indeed, the pulses that implement unitary gates in NISQ devices are stored digitally, and thence the gates cannot be parametrized
continuously, but only using a (possibly infinite) discrete set of values, generalizing floating-point arithmetic (IEEE 754) commonly utilized on classical computers. Formally:

\begin{prob}[Digitized {VQA} minimization problem]\label{p:VQAdigitized}\hfill
  \begin{description}
     \item[Instance] A discrete set of numbers $\mathbb{D}$ representable by a single word on a digital computer.
    An initial state $\ket{\Psi_0}\in \mathbb{C}^n$, a set of generators $\{H_i \}_{i \in \{1,\dots, L\}}\subset \Herm(\mathbb{C}^n)$, where $L$ is the number of layers and an observable $O\subset \Herm(\mathbb{C}^n)$. 
     \item[Optimization version] For $\ket{\Psi(\vec \phi)}\coloneqq U_L(\phi_L)\cdots U_1(\phi_1)\ket{\Psi_0}$ with $U_i(\phi_i)=\e^{-\i H_i \phi_i}$, 
    find $\vec \phi \in \mathbb{D}^L$ that minimizes $\OEx{\vec\phi}\coloneqq~\sandwich{\Psi(\vec \phi)}{O}{\Psi(\vec \phi)}$. 
     \item[Decision version with $a \in \mathbb{R}$]
    For $\ket{\Psi(\vec \phi)}\coloneqq U_L(\phi_L)\cdots U_1(\phi_1)\ket{\Psi_0}$ with $U_i(\phi_i)=\e^{-\i H_i \phi_i}$, 
decide if there exists $\vec \phi \in \mathbb{D}^L$ such that $\OEx{\vec\phi}\coloneqq~\sandwich{\Psi(\vec \phi)}{O}{\Psi(\vec \phi)} \le a$. 
\end{description}
\end{prob}

\begin{figure}[!htb]
    \centering
    \begin{equation}
    \begin{gathered}
D(58,4),D(38,8),D(32,12),D(29,16), \\
D(28,20), D(26,24), D(25,28), D(24,36), \\
D(21,96),D(19,2668), D( 14, 2 \times 10^{5}),  \\
D(13, 6.6 \times 10^{43}), D(12, 1.3 \times 10^{44}), \\ D(11, 4.6 \times 10^{44}),D( 10, 8.6 \times 10^{44}), \\ D(9, 1.6 \times 10^{45}).
\end{gathered}
    \end{equation}
    \caption{A sample of known Universal Diophantine Equations $D(n,d)$, which are polynomials of degree $d$ in $n$ variables. }
    \label{fig:all_UDE}
\end{figure}

In the following, we will present evidence that this is undecidable, using the arguments of Mati{\^a}sevi{\v{c}} \cite{Matiyasevich1993}, presented very nicely by Jones \cite{jones_1982}.
Recall that a set $W$ of integers (representing  strings) is recursively enumerable (r.e.) if there exists a recursive function that can eventually generate any element in $W$. Building on earlier results of Martin Davis, Hilary Putnam, and Julia Robinson, 
Mati{\^a}sevi{\v{c}} has shown
that every recursively enumerable set is Diophantine, that is, for every r.e. set $W$, there is a polynomial $D'_W$ such that for all integers $x$:  
\begin{align}
x \in W \Leftrightarrow \exists v, z_1, z_2, \ldots, z_v \textrm{ s.t. } D'_W(x, z_1, z_2, \ldots, z_v) = 0
\notag 
\end{align}
where $z_1, z_2, \ldots, z_v$ are positive integers and the coefficients of $D'_W$ are integers. 
In this tradition, $z_1, z_2, \ldots, z_v$ are known as variables, and $x, v$ are known as parameters. 
Thus, if one were able to decide the existence of the root of $D'_W$ in positive integers, one could decide the halting problem \cite[Chapter 4]{sipser13} for any Turing machine.

Furthermore, there exist universal integers $n, d$, such that for every r.e. set $W$, there exists its description $v$ such that:  
\begin{align}
x \in W \Leftrightarrow \exists z_1, z_2, \ldots, z_n \textrm{ s.t. } D''(x, v, z_1, z_2, \ldots, z_n) = 0
\notag 
\end{align}
where $D''$ is a polynomial of degree $d$. This polynomial is called a \emph{universal Diophantine equation}.
Here, $x, v$ would still be parameters. 
When deciding on the halting problem \cite[Chapter 4]{sipser13} on the input $x$ for a Turing machine described by a number of parameters, it is convenient to consider additional parameters, as suggested in Figure~\ref{fig:all_UDE}, where the parameters are $u, x, y, z$.

This sum-of-squares construction ensures that the optimization problem can be interpreted as minimizing a squared norm, which aligns naturally with the quantum mechanical objective \eqref{eq:classicalopt11}.

The construction of universal Diophantine equations (UDE) has been the subject of much subsequent work, especially when seeking  
UDE such that the number $n$ of 
variables and the degree $d$ of the polynomial are low.\footnote{
The construction of such polynomials may seem rather obscure, but it is nicely explained in \cite[Note 12.9]{wolfram2002new} including Mathematica code.
} 
This subsequent work \cite{Matiyasevich1993} has established the existence of a variety of universal Diophantine equations $D(n,d)$, as illustrated in Figure~\ref{fig:all_UDE}, cited from \cite{Matiyasevich1993}.
In particular, it is known \cite{Jones1982JONUDE, Matiyasevich1993} that there is a universal Diophantine equation with 58 variables and degree 4, and that solving this $D(58,4)$ is equivalent to solving $D(86\times 10^{43},10)$,
which in turn is equivalent to solving the succinct\footnote{Notice that $D(58,4)$ features many monomials while $D(10,86\times 10^{43})$ features large coefficients.}
 $D(10,2 \times 5^{60})$ displayed in Figure~\ref{fig:all_UDE}.

\subsubsection*{A note on matrix exponentials}

We recall some fundamental properties of matrix exponentials which will be useful for the sequel. Write 
\begin{equation}
    \exp{(kA)} = \e^{k A},
\end{equation}
where $A \in \mathbb{K}^{n\times n}$, where $\mathbb{K}= \mathbb{R}$ or $\mathbb{C}$ usually, and $k$ an element of some ring $R$, usually presented as the ring of real numbers $\mathbb{R}$.  We recall Sylvester's formula:
\begin{theorem}[Sylvester's formula \cite{Sylvester1883}from Caley-Hamilton theorem]\label{th:sylvester}
We can write
\begin{equation}\label{eq:matrixexp}
    \e^{kA} = \sum_{j=0}^{n-1}\varkappa_jA^j.
\end{equation}
where each $\varkappa_j\in \mathbb{K}$.
\end{theorem}
\begin{proof}
This is quite easy to see: let $A \in \mathbb{R}^{2\times 2}$, i.e., $n=2$, without loss of generality and let $a,b \in \mathbb{R}$ be its eigenvalues.
Note that although the original form of the formula is for $k\in\mathbb{R}$ and $A\in\mathbb{R}^{n\times n}$ (for some $n$), the derivation of Sylvester's Theorem can be performed using the Cayley-Hamilton Theorem, which has a generic form appropriate to any commutative ring~\cite[Corollary 4.2.15]{gatto2016hasse}\footnote{This will be quite important in Section \ref{MainResult} when discussing such matrix exponentials. A commutative ring is simply a set together with an operation of addition and multiplication. For example, the set of natural numbers $\mathbb{Z}$ together with the standard multiplication and addition forms a commutative ring. Similarly, the set of $n\times n$ diagonal matrices ${\rm MatDiag}(\mathbb{C}^n)$ with values in complex numbers, together with the standard matrix addition and matrix multiplication, forms a commutative ring.}. See also~\cite{merzbacher1968matrix}. We assume for simplicity that 
$A$ is diagonalizable. In our setting, the matrices of interest (e.g., Hermitian Hamiltonians) are always diagonalizable, so this assumption holds.
\end{proof}
The eigendecomposition of a $A$ is $A = Q\left( \begin{smallmatrix} p  & 0 \\ 0 &q   \end{smallmatrix}\right)Q^{-1}$ for an invertible square matrix $Q$. Then, $\e^{A} = Q\left( \begin{smallmatrix} \e^p  & 0 \\ 0 &\e^q   \end{smallmatrix}\right)Q^{-1}$. Finally, observe that by finding such a $Q$ we arrive at
$\e^A = (q-p)^{-1}[(q\e^p-p\e^q)\mathds{1}+(\e^q-\e^p)A]$. This generalizes to a higher $n$.

In Eq. \eqref{eq:matrixexp} the coefficients in the finite series expansion are given as functions of the eigenvalues of $A$,
\begin{align}\label{eq:matrixexp2}
     \e^{k\lambda_i } &= \sum_{j=0}^{n-1} \varkappa_j \lambda_i^j \\
                      &= \varkappa_0+\varkappa_1 \lambda_i + \varkappa_2 \lambda_i^2 + \ldots + \varkappa_{n-1}\lambda_i^{n-1}
\end{align}

To understand this better let us consider a concrete example where $A \in \mathbb{C}^{2\times 2}$ given as
\begin{equation}
    A = \begin{pmatrix}
       a&b \\
       c&d \
        \end{pmatrix},
\end{equation}
and let $k \in \mathbb{R}$. The characteristic polynomial of $A$ is 
\begin{equation}
    p(x) = x^2- (a+d) x  + ad - bc,
\end{equation}
with eigenvalues
\begin{equation}
    \begin{aligned}
        x_{1,2} & =  \tfrac{1}{2} (a + d \mp \sqrt{a^2 + 4 b c - 2 a d + d^2}). 
    \end{aligned}
\end{equation}
Using Eq. \eqref{eq:matrixexp2} we can compute the coefficients $\{\varkappa_j\}$ appearing in Eq. \eqref{eq:matrixexp} using
\begin{equation}\label{eq:exponentials}
    \begin{aligned}
        \e^{kx_1} &=  \varkappa_0  + \varkappa_1 x_1 \\
        \e^{kx_2} &= \varkappa_0 +  \varkappa_1 x_2.
    \end{aligned}
\end{equation}
to obtain 
\begin{align}\label{eq:vak}
    \varkappa_0(x_1,x_2,k) &= \frac{ \e^{kx_2}x_1-\e^{kx_1}x_2}{x_1 - x_2} \\
    \varkappa_1(x_1,x_2,k) &= \frac{\e^{kx_1} - \e^{kx_2}}{x_1- x_2}.
\end{align}
From Eq. \eqref{eq:vak}, we can write $\varkappa_0,\varkappa_1$ as functions of the eigenvalues $x_1,x_2$ using Eq. \eqref{eq:matrixexp}. We obtain:
\begin{equation}\label{eq:matrixexp3}
    \e^{kA} = \varkappa_0(x_1,x_2,k) \mathds{1} + \varkappa_1(x_1,x_2,k) A.
\end{equation}


Observe that when $k \in  \i \mathbb{R}$, each coefficient $\varkappa_j \in \mathbb{C}$. These coefficients are determined entirely by the eigenvalues $\left\{\lambda_i\right\}$ of $A$, and depend on the scalar $k$ as well. They are therefore not independent parameters, but functions of the spectrum of $A$.
This observation is important: for a matrix $A \in \mathbb{K}^{n \times n}$ with $n^2$ (real) degrees of freedom, solving an equation of the form

$$
\operatorname{Tr}\left[A \left(\varkappa_j B\right)\right]=0
$$
where $B \in \mathbb{K}^{n \times n}$, does not introduce new degrees of freedom-since $\varkappa_j$ is not a free variable but a fixed scalar determined by $A$ and $k$.


\section{The Undecidability Conjecture}\label{MainResult}

We would like to encode the universal Diophantine equation $D(58, 4)$ into the digitized {VQA} minimization problem.
That is, we wish to take one universal Diophantine equation, $D(58, 4)$, which is a system of polynomial equations,
but which can be easily transformed into a single sum-of-squares polynomial $D(58, 8)$. This sum-of-squares version can be seen as:
\begin{equation}\label{eq:classicalopt10}
    \min\limits_{\vec{\phi}\in\mathbb{D}^{58}}\,  \sum_j q_j(\vec{\phi})^2,
\end{equation}
where $q_j$ denotes one summand of the sum-of-squares form of the UDE $D(58, 8)$. This \textit{sum-of-squares} construction ensures that the optimization problem can be interpreted as minimizing a squared norm, which aligns naturally with the quantum mechanical objective~\eqref{eq:classicalopt11}.

Now, let us consider the $L=58$ and $n=5$ case of Problem~\ref{p:VQA}:
\begin{align}\label{eq:classicalopt11}
    \min_{\vec{\phi}\in\mathbb{D}^{58}} \, \left\| O \prod_{i=0}^{57} U_{58-i}(\phi_{58-i})\ket{\Psi_0} \right\|^2
\end{align}
and let us show the existence of the Hermitian operator $O \in \Herm(\mathbb{C}^n)$, as well as the state $\ket{\Psi_0} \in \mathcal{H}$, and unitaries $ U_{58}(\phi_{58}), \ldots, U_1(\phi_1) \in U(n)$ such that 
the classically-solved optimization problem \eqref{eq:classicalopt11} 
in the discrete set $\mathbb{D}^{58}$ of numbers representable by 58 words in a digital computer,
attains objective value zero if and only if there is a zero of the UDE $D(58, 4)$ in $\mathbb{D}^{58}$. In other words, proving Theorem \ref{thm:digivqeundecidable} amounts to finding a map \eqref{eq:classicalopt11} $\to$ \eqref{eq:classicalopt10} and the existence of such an encoding map boils down to the nontrivial task of proving the existence of a common zero of a system of complex polynomials. The challenge lies in ensuring that the variational quantum circuit evaluates to zero if and only if the UDE has a solution—a condition that hinges on the ability to construct appropriate Hamiltonians, a state, and an observable to match the algebraic structure of the UDE.

The conversion involves several \emph{layers of variables} that we qualitatively describe now, and further describe in Table \ref{Table:Variables} in the Supplementary material: 

\begin{itemize}
    \item Within any UDE, there are parameters that encode the Turing machine.\footnote{In the example of Figure \ref{fig:UDE}, this would be $u, v, x, z$.}
    \item Within the same UDE, there will be a vector of variables restricted to integers; \footnote{
    In the example of Figure \ref{fig:UDE}, these would be $a, b, c, \ldots, s, t, w, \alpha, \gamma, \nu, \theta, \lambda, \tau, \phi$. Notice, however, that we override the meaning of some of these symbols in this section.} in $D(58, 4)$, this is a 58-vector, which we denote $\vec \phi$. 
    \item Ultimately, we can decide on the state vector $\ket{\Psi_0}$, and the operators $O, U_{58}(\phi_{58}), \ldots, U_1(\phi_1)$. In this system, $\ket{\Psi_0}$ will enjoy some freedom, which will be discussed later. Note that the dimensions of these objects are restricted by the dimension of the UDE\footnote{By Sylvester's formula.}.
    \item However, the existence of an encoding map is determined by the existence of a non-empty zero set $V(I)$ of simultaneous zeros of a certain system of complex polynomials closely associated with the probability amplitude of $O$.
\end{itemize}

First, let us recall that it is sufficient to consider only  universal Diophantine equations that are sum-of-squares polynomials, which is compatible with Prob. \eqref{eq:classicalopt10}. Fortunately, there are a number of those, as exemplified in Figure \ref{fig:UDE}.
The sum-of-squares nature of many universal Diophantine equations is due to the fact that those are usually devised as systems of polynomial equations, and only at the end one obtains a single polynomial by the sum-of-squares construction.

\noindent
\emph{Assumptions}. Without loss of generality of the proof, we will assume that $[H_i,H_j]=0$ for all $H_i,H_j \in \Herm(\mathbb{C}^n)$. This assumption will allow us to efficiently make use of Sylvester's formula \eqref{eq:matrixexp}. Furthermore, before we proceed, we need to note that when taking into account Sylverster's formula we directly get a cut-off on the rank of the generators $\{H_i \}_{i=1}^{58}$ which now are elements of $\Herm(\mathbb{C}^5)$. Similarly, this enforces $\ket{\Psi_0} \in \mathbb{CP}^5$. 

With $U_i(\phi_i)=e^{-\i\phi_iH_i}$ as in Eq. \eqref{Eq:1}, we have
\begin{align}\label{eq:expandus}
 \prod_{i=0}^{57} U_{58-i}(\phi_{58-i}) & = \prod_{i=0}^{57} \e^{-\i\phi_{58-i}H_{58-i}} \\ 
 &  = \e^{-\i\sum\limits_{i=1}^{58}\phi_{i}H_{i}} \\ \label{eq:expandus3}
 & = \sum\limits_{j=0}^4 \varkappa_j \left(\sum\limits_{i=1}^{58}\phi_{i}H_{i}\right)^j \\  \label{eq:cb}
 & = \vec c\cdot\vec b(\vec\phi) = \sum_{a=1}^m c_a b_a(\vec \phi).
\end{align}
Here, we applied Sylvester's Theorem~\ref{th:sylvester} to obtain the third line, Eq. \eqref{eq:expandus3}, with the expansion up to $4$ since $H_i\in\Herm(\mathbb{C}^{5})$. The index $a$ is a multi-index that labels the elements of $\vec c$ and $\vec b$ according to the degree of the monomial in the elements of the variational parameter vector $\vec \phi$, see for example Eq. \eqref{example} below. The monomial basis vector $\vec{b}(\vec{\phi})$ encodes the possible set of monomial terms in the VQA (i.e., $\phi_1^2\phi_{36}$, etc.), and has dimension $m = \tfrac{62!}{58!4!} \sim \mathcal{O}(10^5)$ (generically $\tfrac{(L+n-1)!}{L!(n-1)!}$). Each component of the $m$-dimensional vector $\vec c$ is itself an element of $\Herm(\mathbb{C}^{5})$ and contains the coefficients that multiply each monomial term of $\vec{b}(\vec \phi)$, a product of the corresponding complex number $\varkappa_j$ with a certain monomial in the generators $\{H_i\}_{i=1}^{58}$. Note that the set of complex numbers $\{ \varkappa_j\}_{j=0}^4$ is not independent. They can be described as functions of the variational parameters and the eigenvalues of the Hamiltonians by solving for them the system
\begin{equation}
    \begin{aligned}
    \e^{-\i x_1} &= \sum_{l=0}^{4}\varkappa_l x_1^l \\
    \e^{-\i x_2} &= \sum_{l=0}^{4}\varkappa_l x_2^l \\
    & \vdots \\
    \e^{-\i x_5} &= \sum_{l=0}^{4}\varkappa_l x_5^l, \
\end{aligned}
\end{equation}
as in Eq. \eqref{eq:exponentials}, where 
\begin{align}
    x_i = \sum_{j=1}^{58} (\phi_jH_{j})_{ii},
\end{align}
that is, each $x_i$ is given as the sum of the eigenvalues of the $i$-th row of all Hamiltonians.

In order to clearly understand the structure of the vectors $\vec c$ and $\vec{b}(\vec \phi)$ we provide a simple example in the Supplementary Material, page \pageref{sec:example2}.

With this new description of $\prod_{i=0}^{57} U_{58-i}(\phi_{58-i})$ as the product $\vec c\cdot\vec b(\vec\phi)$ we can insert the latter into the VQA objective~\eqref{eq:classicalopt11}, with the aim of reaching \eqref{eq:classicalopt10}, and seek to enforce the corresponding equation,
\begin{equation}
   \begin{aligned}
    \left\| O \vec c\cdot\vec b(\vec\phi) \ket{\Psi_0}\right\|^2 &=  \sum\limits_j q_j(\vec\phi)^2 \\
    & = \sum_{j} c^U_j b^U_j(\vec{\phi}) \\
    & = \vec c^U \vec b^U(\vec\phi),
    \end{aligned} 
\end{equation}
where $\vec c^U\in\mathbb{Z}^{m^2}$ and $\vec b^U$ are the coefficients and corresponding monomial basis vectors of the sum-of-squares form of the UDE, i.e., the decomposition of Eq. \eqref{eq:classicalopt10} as $\vec{c}^U \vec{b}^U$. This equation is equivalent to a set of equations on the coefficients, i.e., 
\begin{align}\label{eq:newsystem}
    \braket{\Psi_0 |  (O\vec c)^\dagger O\vec c | \Psi_0 }_j =  \vec c^U_j
\end{align}
where $j$ is selecting the corresponding monomial term of $(O\vec c)^\dagger O\vec c$ of $\vec b^U$. In other words, the subscript $j$ denotes projection onto the monomial basis term $b_j^U(\vec{\phi})$. That is, we extract the coefficient of the $j$-the monomial in the expansion of the squared amplitude.

Let $\vec{x}$ denote the vector of entries of the operator $O$ and let $\vec{w}$ denote the entries of $\ket{\Psi_0}$. The system of equations \eqref{eq:newsystem} corresponds to $2m$ equations $\{ g_i(\vec x,\vec w) = 0 \}_{i=1}^{m^2}$ to solve for. 

\begin{conj}[Solution of a Highly Structured System]\label{conj:system}
    For the set of coefficients $\vec c^U_j$ associated with the Diophantine polynomial system with 58 variables, there exists at least one pair of $O$ and $\ket{\Psi_0}$ such that there exists at least one solution $\vec c^U$ to Sys. \eqref{eq:newsystem}.
\end{conj}

Essentially, system \eqref{eq:newsystem} corresponds to finding the zero-set $V(I)$ of a system of complex polynomial equations $\{ g_i = 0 \}_{i=1}^{2m}$ where $I$ is the radical formed out of them. 
By Hilbert's Nullstellensatz \cite{Cox2015}, a solution of a system of polynomials $\{ g_1=0, \ldots, g_{2m} = 0\}$ with complex coefficients in $l$ variables has no solution in $\mathbb{C}^l$ if and only if one can find a set of Bézout polynomials $\{g'_1= 0, \ldots, g'_k = 0 \}$ such that 
\begin{equation}
\label{eq:Hilbert}
    \sum_{i=1}^k g_i g'_i = 1,
\end{equation}
i.e., we require that the set $\{g_1,\ldots,g_k\}$ forms a proper ideal in $\mathbb{C}[x_1,\ldots, x_l]$.
Using the technique of Sombra \cite{sombra1999sparse}, an upper bound on the total degree of the $g'_i$ can be computed by solving an exponentially large linear system.
Alternatively, explicit solutions could be obtained by \cite{duchi2019solving}.

\begin{theorem}[Undecidability of VQAs]
\label{thm:digivqeundecidable}
Assuming Conjecture \ref{conj:system} (``Solution of a Highly Structured System'') holds, the decision version of digitized {VQA} minimization (Problem~\ref{p:VQAdigitized}) is undecidable for $L=58$ layers. 
\end{theorem}

Theorem \ref{thm:digivqeundecidable} then implies that there is no finite-time algorithm to solve the problem, that is,

\begin{coro}[Uncomputability of VQAs]\label{cor:VQE}
Assuming Conjecture \ref{conj:system} (``Solution of a Highly Structured System'') holds, there exists no recursive function to decide digitized {VQA} minimization in the decision version (Problem~\ref{p:VQAdigitized}) for $L=58$ layers.
Consequently, the optimization version is uncomputable for $L=58$ layers.
\end{coro}

This result is conditional on the validity of the Conjecture \ref{conj:system}. While we do not proceed to prove that this system indeed contains at least one solution, there are various reasons for this to be true. (See above.) 
Note that there are a number of Diophantine equations that we can consider, with increasing degree and variable dimension. From the standpoint of the degrees of freedom arithmetic described in the following, however, it intuitively appears to be the most agreeable equation to the formulated conjecture, and so we formulate it as such. Of course, the analogous conjecture for any polynomial Diophantine system would be similarly informative in regard to VQA.

\emph{Degrees of freedom}. For each Hermitian matrix $H \in \Herm(\mathbb{C}^5)$, we have $5$ real degrees of freedom from the diagonals (which are real) and $2\times 10$ real degrees of freedom ($2\times 5$ complex numbers) from the off-diagonals, summing to 25 real degrees of freedom. Any operator $O \in \Herm(\mathbb{C}^5)$ must be described by the $5^2-1$ generators $\{\sigma_i\}_{i=1}^{25}$ of the corresponding Lie algebra as $O = \sum_{i=1}^{25}\alpha_i\sigma_i$ matching the number of parameters.  The degrees of freedom of $\vec c$ is equivalent to the number of eigenvalues of the diagonal generators $\{H_i \}_{i=1}^{58}$, that is $58 \times 5 = 290$. Once a combination is chosen, all elements of $\vec c $ are fixed. Finally, naively one might think that $\ket{\Psi_0} \in \mathbb{C}^5$ but given the trace condition a complex degree of freedom is removed, while the same holds for a global phase, thus one is left with 8 real degrees of freedom. This matches the dimension of $\mathbb{CP}^5$ where the set of states such as $\ket{\Psi_0}$ belongs to. We summarize the degrees of freedom for our system in Table \ref{table:dof}.
\vspace{1em}

\begin{table}[t!]
\begin{tabular}{llc} \toprule[2pt]
\multirow{2}{*}{Object} & \multirow{2}{*}{Space} & \multirow{2}{*}{\begin{tabular}[c]{@{}l@{}}Degrees of \\ freedom (real)\end{tabular}} \\
        &                        &  \\ \hline
$\ket{\Psi_0}$  & $Q_5 \coloneqq \mathbb{CP}^5$         & 8\\
$O$     &  $\Herm(\mathbb{C}^5)$         &  25 \\
$H_i$     &  $\Herm(\mathbb{C}^5)$         &  5 \\
$\{H_i\}$     &  $\Herm(\mathbb{C}^5)$         &  $5\times 58=290$ \\
\hline
\end{tabular}
\caption{The space of (pure) states of 5-dits, denoted as $Q_5$ is a proper subset of the $d^{2}-1$, $d=5$, dimensional Hilbert-Schmidt sphere (ball, for mixed states). The corresponding Bloch manifold for $\{ \ket{\Psi_0} \}_{\rm pure} \cup \{ \ket{\Psi_0} \}_{\rm mixed}$ is ${\rm SU}(5)$. For $\{ \ket{\Psi_0} \}_{\rm pure}$ the number of degrees of freedom is $2d-2$, the real dimension of $Q_5$ \cite[Section 8]{bengtsson2006}.
Finally, for $c_{(i,j)} \in \Herm(\mathbb{C}^5)$, we do not obtain new degrees of freedom on top of the degrees of freedom in $\vec{\phi}$ and $\{H_i\}$.
}
\label{table:dof}
\end{table}

\paragraph*{An illustration of encoding a sum-of-squares polynomial into small VQA}\label{Encoding}
As an \emph{illustration} on how to encode a universal Diophantine equation polynomial into VQA let us consider a simple example of how a generic second-degree sum-of-squares polynomial can be written in the form~\eqref{eq:classicalopt1}. Consider the sum-of-squares polynomial 
\begin{align}\label{eq:examplepoly}
    p(x,y)=\eta_x x^2 + \eta_y y^2+\eta_{xy} x^2 y^2,
\end{align}
with coefficients $\eta_x,\eta_y,\eta_{xy} \in \mathbb{Z}, \mathbb{R}$ or $\mathbb{C}$ that are given and fixed (i.e., not free). The individual terms $\{q_j\}$ of the UDE have as their highest total power 2. Thus, we consider that the state dimension is $3$ in order to generate the required powers in the use of Sylvester's formula. 

We wish to construct the corresponding:  
\begin{align}\label{eq:example}
    \left\|  OU_x(x)U_y(y) |\Psi_0 \rangle \right\|^2.
\end{align} 
Let $H_x = X, H_y = Y$. Recall that
\begin{align} \label{eq:commute}
    U_x(x)U_y(y) & =\e^{-\i X} \e^{-\i Y}=\e^{-\i(X+Y)}
    \\ & = \sum\limits_{j=0}^{2} \varkappa_j (xX+yY)^j
    \\ \nonumber & =  c(0,0)+c(0,1)x+ c(1,0) y  \\ &\quad +c(1,1) xy + c(2,0)x^2 +c(0,2) y^2 \\
    & = \vec c \cdot \vec b(x,y)
\end{align}
where we have applied Sylvester's Theorem~\ref{th:sylvester} at the second equality, noting that $X,Y\in\Herm(\mathbb{C}^3)$ and thus there are terms up to degree two in the polynomial appearing in its application. Here we require that $[X,Y]=0$ in order to proceed to the second equality of Eq. \eqref{eq:commute}. Explicitly, we have
\begin{align*}
    c(0,0) & = \varkappa_0\mathds{1} \\
    c(1,0) & = \varkappa_1 X \\
    c(0,1) & = \varkappa_1 Y \\
    c(2,0) & = \varkappa_2 X^2 \\ 
    c(0,2) &= \varkappa_2 Y^2  \\
    c(1,1) &= 2\varkappa_2 XY. \\
\end{align*}
Note the dot product of the monomial vector $\vec b(x,y)=\begin{pmatrix} 1 & x & y & xy & x^2 & y^2 \end{pmatrix}^T$ with the vector of coefficients $\vec c$. Applying $O$ from the left, that is, in $OU_x(x)U_y(y)$ we obtain the modified polynomial,

\begin{equation}
  \begin{aligned}
    O U_x(x)U_y(y) & =   Oc(2,0)x^2  +Oc(0,2) y^2  \\ &\quad + Oc(1,1) xy +Oc(0,0)  \\
                   & \quad +Oc(1,0)x+ Oc(0,1) y. \
\end{aligned}  
\end{equation}
Thus, when we substitute this expression into the VQA objective we obtain,
\begin{align}\label{eq:almost}
    &\left\|  OU_x(x)U_y(y) |\Psi_0 \rangle \right\|^2 = \left\|O (\vec c \cdot \vec{b}(x,y))|\Psi_0 \right\|^2 \\
    &= \langle \Psi_0|\sum_{\substack{0\le j^1_1+j^2_1\le 2\\0\le j^1_2+j^2_2\le 2}} \tilde{O}_{j^1_1,j^2_1}^\dagger \tilde{O}_{j^1_2,j^2_2} x^{j^1_1+j^1_2}y^{j^2_1+j^2_2} |\Psi_0 \rangle
\end{align}
where $\tilde{O}_{j^1,j^2}=Oc(j^1,j^2) \in \Herm(\mathbb{C^5})$, since, as we have seen, each $c(i,j)$ corresponds to a certain monomial in Hamiltonians $\{H_i\}$. Expanding yields

\begin{equation}\label{eq:system}
    \begin{aligned}
    0 &= \braket{\Psi_0|\tilde{O}_{0,0}^\dagger \tilde{O}_{0,0}|\Psi_0} \\
    0 &= \braket{\Psi_0|\tilde{O}_{1,0}^\dagger \tilde{O}_{0,0}+\tilde{O}_{0,0}^\dagger \mathcal{O}_{1,0}|\Psi_0} (= 0\cdot y)\\
    0 &= \braket{\Psi_0|\tilde{O}_{0,1}^\dagger \tilde{O}_{0,0}+\tilde{O}_{0,0}^\dagger \tilde{O}_{0,1}|\Psi_0} (= 0 \cdot x)\\
    \eta_{xy} &= \braket{\Psi_0|\tilde{O}_{0,0}^\dagger \tilde{O}_{1,1}+\tilde{O}_{0,1}^\dagger \tilde{O}_{1,0}+\tilde{O}_{1,0}^\dagger \tilde{O}_{0,1}|\Psi_0} \\
    \eta_y &= \braket{\Psi_0|\tilde{O}_{2,0}^\dagger \tilde{O}_{0,0}+\tilde{O}_{0,0}^\dagger \tilde{O}_{2,0}+\tilde{O}_{1,0}^\dagger \tilde{O}_{1,0}|\Psi_0} \\
    \eta_x &= \braket{\Psi_0|\tilde{O}_{0,2}^\dagger \tilde{O}_{0,0}+\tilde{O}_{0,0}^\dagger \tilde{O}_{0,2}+\tilde{O}_{0,1}^\dagger \tilde{O}_{0,1}|\Psi_0} .
\end{aligned}
\end{equation}
If there exist $O,H_1,H_2$ and $\ket{\Psi_0}$ such that~\eqref{eq:system} has a solution, then we can encode the problem of finding zeros of the polynomial as this VQA. But in fact, we have two degrees of freedom in $\ket{\Psi_0}$ and six degrees of freedom with $O$. We potentially have at least two additional degrees of freedom with the choice of $H_1$ and $H_2$, which, in turn, modify $\vec c$.  In Fig. \ref{fig:sage} we provide a sample SAGE script to test whether the space we are interested in is empty or not. 

\paragraph*{Summary of the construction} 

We started from Eq. \eqref{eq:example} and applied Sylvester's formula (based on Caley-Hamilton theorem) to the product of unitaries $\prod_{i=0}^{57} U_{58-i}(\phi_{58-i})$ as to obtain $\vec c \cdot \vec b(\vec \phi)$. This allows us to reformulate the original product of unitaries, as usually understood in the VQAs, in terms of a system of complex equations which when enforced to produce the coefficients of the square of the UDE of interest, $D(58,4)$, achieve the encoding of the VQA into the UDE provided at least one solution exists.


\section{Translating the Main Result to QAOA} 
The reasoning above and most importantly Theorem \ref{thm:digivqeundecidable} and Corollary \ref{cor:VQE} apply equally well to {QAOA} in the digitized version:

\begin{prob}[{QAOA} minimization problem, \cite{bittel2021training}]\label{p:QAOA}\hfill
	\begin{description}[noitemsep,leftmargin=0.5cm,font=\normalfont]
		\item[Instance] Two Hamiltonians $H_b,H_c\in \Herm(\mathbb{C}^n)$ and the number of layers $L$ in unary notation\footnote{This means that the length of the input scales linearly with $L$.}.
		
		\item[Optimization version] For a tunable state $\ket{\Psi(\vec \beta, \vec \gamma)}\coloneqq U_b(\beta_L)U_c(\gamma_L)\cdots U_b(\beta_1)U_c(\gamma_1) \ket{\Psi_0}$, where $\ket{\Psi_0}$ is the ground state of $H_b$, $U_b(\beta)=\e^{-\i H_b \beta}$ and $U_c(\gamma)=\e^{-\i H_c \gamma}$,
		find $\vec \beta,\vec \gamma \in \mathbb{R}^d$ which minimize $\OEx{\vec \beta, \vec \gamma}\coloneqq\sandwich{\Psi(\vec \beta, \vec \gamma)}{H_c}{\Psi(\vec \beta, \vec \gamma)}$. 
		
\item[Decision version for $a \in \mathbb{R}$] 
For a tunable state $\ket{\Psi(\vec \beta, \vec \gamma)}\coloneqq U_b(\beta_L)U_c(\gamma_L)\cdots U_b(\beta_1)U_c(\gamma_1) \ket{\Psi_0}$, where $\ket{\Psi_0}$ is the ground state of $H_b$, $U_b(\beta)=\e^{-\i H_b \beta}$ and $U_c(\gamma)=\e^{-\i H_c \gamma}$,
		decide if there exists  $\vec \beta,\vec \gamma \in \mathbb{R}^d$ such that $\OEx{\vec \beta, \vec \gamma}\coloneqq\sandwich{\Psi(\vec \beta, \vec \gamma)}{H_c}{\Psi(\vec \beta, \vec \gamma)} \le a$.
		
	\end{description}
\end{prob}

\begin{theorem}[Undecidability of optimization in {QAOA}] \label{thm:QAOAundecidable}
	Decision version of {QAOA} minimization problem (Problem~\ref{p:QAOA}) is undecidable for $L=58$ layers.
\end{theorem}

\begin{proof}
Consider $L$ unitary operators parametrized by $\vec{\gamma} \in \mathbb{R}^L$ acting on some prepared state $\ket{\Psi(\vec{\beta})}$ which is a function of $\vec{\beta} \in \mathbb{R}^L$. Then for any $L$ we define the VQA ansatz with parameter $\vec{\gamma}$
    \begin{align}
        \ket{\Psi(\vec{\gamma})} = \prod_{i=0}^{L-1} U_c(\gamma_{L-i}) \ket{\tilde{\Psi}_0(\vec{\beta})},
    \end{align}
where the VQA layers of unitaries act on a parameterized state 
    \begin{align}
        \ket{\tilde{\Psi}_0(\vec{\beta})} = \prod_{i=0}^{L-1} U_b(\beta_{L-i})\ket{\Psi_0},
    \end{align}
and, as before, $U_b \eqqcolon {\rm e}^{-\i H_b}$ and $U_c \eqqcolon {\rm e}^{-\i H_c}$.
Assume $H_b$ is diagonal, as in Eq. \eqref{Hb} below or by diagonalizing the usual tensor product of $\sigma_x$ matrices as in \cite{farhi2014quantum}, and $H_c$ is symmetric. Since $H_b$ and $H_c$ commute and we form the QAOA ansatz
	\begin{align*}
	    \ket{\Psi(\vec{\beta}, \vec{\gamma})} &= \prod_{i=0}^{L-1} U_c(\gamma_{L-i}) U_b(\beta_{L-i})\ket{\Psi_0} 
	\end{align*}
    and Theorem \ref{thm:digivqeundecidable} applies directly. 
    Notice that to show undecidability, it is sufficient to show the undecidability of this special case. 
    
    In general, QAOA utilizes the Ising model Hamiltonian which is proportional to $\otimes^N \sigma_z$, and corresponds to a symmetric block matrix indeed. This means that $U_b(\beta_{i+1})U_c(\gamma_i) = U_c(\gamma_i)U_b(\beta_{i+1})$. Then Theorem \ref{thm:digivqeundecidable} applies. 
	
	Alternatively, one can use the reduction from single layer VQA to QAOA, which has been introduced by Bittel and Kliesch \cite{bittel2021training},
	and then invoke our Theorem \ref{thm:digivqeundecidable} again.
\end{proof}

	For the convenience of the reader, the reduction of Bittel and Kliesch \cite{bittel2021training} considers a Hilbert space $\H = \mathbb{C}^{2d+1}$ and considers a mixer Hamiltonian $H_b$ that takes the form
	\begin{align}\label{Hb}
	    H_b = \begin{psmallmatrix}
	           E_1& & & & &\\
	              &-E_1& & & &\\
	              & & \ddots & & &\\
	              & & & E_d & & \\
	              & & &  & -E_d&\\
	              & & &  &  & 1  \
	    \end{psmallmatrix}
	\end{align}
	where $|E_i|<1$ for all $i\in[d]$. For the optimization Hamiltonian $H_c$ (which amounts to the VQE Hamiltonian), consider
    \begin{align}
        H_c&=O\oplus 0 +\tau\left(\ket{+}_{2d}\bra{2d+1}+\ket{2d+1}\bra{+}_{2d}\right), 
    \end{align}
	where $\tau \in \mathbb{R}$ is some constant to be fixed, $\ket{+_{2d}}=\sum_{j=1}^{2d}\ket{j}/\sqrt{2d}$. The observable $O$ is defined as follows:
	
	\begin{align}
	    O_{i, j}=\begin{cases}
	        O_{i, j}^{\prime} & i \neq j, \\
	        -\sum_{\alpha=1}^{2 d} O_{\alpha, j}^{\prime} & i=j,
	    \end{cases}
	\end{align}
	where $O^{\prime}:=\frac{d}{8}A\otimes \begin{psmallmatrix}1&1\\1& 1\end{psmallmatrix}$
	and $A\in \{0,1 \}^{d\times d}$. Note that
	$O\oplus 0$ refers to $O$ being embedded in the first $2d$-dimensional summand of $\mathcal{H}$. The ground state energy of the mixer Hamiltonian $H_b$ is  $\lambda_{\min}(H_b)=-1$ as deducted from Eq. \eqref{Hb} with the ground state being $\ket{2d+1}$. Then, applying the optimization Hamiltonian and using the fact that $\ket{+_{2d}}$ is an eigenstate of $O\oplus 0$ yields
    \begin{equation}
	\begin{aligned}
	\ket{\Psi(\gamma)}&\coloneqq U_c(\gamma)\ket{2d+1}\\
	&=\cos(\tau\gamma)\ket{2d+1}+\i\sin(\tau\gamma)\ket{+_{2d}}.\
		\end{aligned}
	\end{equation}
	Overall, the \emph{variational state} $\ket{\Psi(\beta,\gamma)}$ can be written as 
	\begin{equation*}
	\begin{aligned}
	\ket{\Psi(\beta,\gamma)}&=U_b(\beta)U_c(\gamma)\ket{2d+1}\\
	&=
	\cos(\tau\gamma)\e^{\i\beta}\ket{2d+1}\\&\quad +\i\sin(\tau\gamma)\tfrac{1}{\sqrt{2d}}\sum_{j=1}^d\e^{-\i E_j\beta}\ket{2j-1}+\e^{\i E_j\beta}\ket{2j}\, ,
	\end{aligned}
    \end{equation*}
	which we use to derive an expression for the expectation value
	\begin{equation}\nonumber
		\begin{aligned}
	\OEx{\beta,\gamma}&=\bra{\Psi(\beta,\gamma)} H_c\ket{\Psi(\beta,\gamma)}\\
	&=\sin^2(\tau\gamma)f(\beta)+2\tau\cos(\tau\gamma)\sin(\tau\gamma)g(\beta)
	\end{aligned}
	\end{equation}
	with
	\begin{align}\label{eq:O_QAOA1}
	f(\beta)&=\frac{1}{4}\sum_{i,j}(A_{i,j}(\cos(E_i\beta)\cos(E_j\beta)-1)\, ,\\
	g(\beta)&=-\frac{\sin(\beta)}{d}\sum_{i=1}^d\cos(E_j\beta)\,.
	\end{align}
	For $\tau\ll1$ the contribution of $g$ becomes insignificant and $\gamma=\frac{\pi}{2\tau}$ minimizes the objective function as $f(\beta)\leq0$.


\section{Further results on QAOA for open quantum systems}

In this Section, we remark on how the consideration of \emph{open} quantum systems can influence undecideability, with respect to the conjecture presented. In open quantum systems, the system can be modeled as
$U_b(\beta)=\e^{-\i H_b \beta} + \eta_b$ and $U_c(\gamma)=\e^{-\i H_c \gamma} + \eta_c$, for which unitarity is not assured for all 
realizations of noise $\eta_b,\eta_c\in \Herm(\mathbb{C}^n)$. 
For example, consider a level two system with $\eta_b$ given as:
\begin{align}
  \eta_b =   \left(\begin{array}{c c}
z & w \\
-\e^{\i (\theta+\epsilon)} w^{*} & \e^{\i \theta} z^{*}
\end{array}\right),
\end{align}
where $z,w\in \mathbb{C}$, $^*$ denotes complex conjugation and $\theta,\epsilon \in \mathbb{R}$. Even for $0<\epsilon \ll 1$, $\eta_b$ is nonunitary. One can impose $\epsilon$ to be a monotonically increasing function of time, $\epsilon(t)$ with initial condition $\epsilon(0)=0$, such that for small $t$ the evolution can be written as:
\begin{align}
    |\Psi(t)\rangle = (\e^{-\i H t} + \eta_b(t))|\vec 0\rangle.
\end{align}
In the context of QAOA, the parameter $t$ can be substituted with the angle $\vec \beta$.
Notice that this is a perfectly reasonable assumption: {QAOA} is precisely motivated as a variational solver to perform on noisy quantum devices. 

Another possibility is to consider controlled systems such as states that evolve as
\begin{align*}
    \i \frac{d}{dt}|{\Psi(t)}\rangle=\left(H_{}+\epsilon(t) \xi\right) |\vec 0\rangle ,\left.\quad |\Psi\right\rangle|_{t=0}=|\vec{0}\rangle.
\end{align*}
This state evolves according to Hamiltonian $H \in {\rm Herm}(\mathbb{C}^n)$ as well as a Hermitian operator $\xi\in {\rm Herm}(\mathbb{C}^n)$ controlled by a time-dependent external field $\epsilon(t)$. In this context, if $H$ is the target Hamiltonian $H_c$, we see that the addition of the second summand $\epsilon(t) \xi$ makes the evolution operator nonunitary.

Clearly, when one does not have access to the samples of noise, it is impossible to say \emph{a priori} what would be evolution of the 
state under, e.g., {QAOA}. 
Interestingly, however, even in the \emph{clairvoyant} setting, where we knew the realization of the noise \emph{a priori},
and even in the case where the distribution of the noise were singular, i.e., the realizations $\eta_b,\eta_c\in \Herm(\mathbb{C}^n)$ 
of noise were constant throughout time,
possibly very small but sufficient to make the evolution non-unitary, 
it is impossible to decide whether the non-unitary evolution is convergent or not. 


\begin{prob}[{QAOA} convergence]\label{p:QAOAstability}\hfill
	\begin{description}[noitemsep,leftmargin=0.5cm,font=\normalfont]
		\item[Instance] Two Hamiltonians $H_b,H_c\in \Herm(\mathbb{C}^n)$ and the number of layers $L$ in unary notation.
		Two realizations of noise associated with the two Hamiltonians $\eta_b,\eta_c \in \Herm(\mathbb{C}^n)$, which assure that 
		neither $U_b(\beta)=\e^{-\i H_b \beta} + \eta_b$ nor $U_c(\gamma)=\e^{-\i H_c \gamma} + \eta_c$ are unitary,
		
		\item[Decision version] For a tunable state $\ket{\Psi(\vec \beta, \vec \gamma)}\coloneqq U_b(\beta_L)U_c(\gamma_L)\cdots U_b(\beta_1)U_c(\gamma_1) \ket{\Psi_0}$, where $\ket{\Psi_0}$ is the ground state of $H_b$,
		decide whether there exists a limit to the sequence produced by {QAOA}. 
	\end{description}
\end{prob}

\begin{theorem}
\label{thm:qaoaconvundecidable}
The {QAOA} convergence problem is undecidable. 
\end{theorem}

\begin{proof}
We can prove this by reduction from \cite{blondel2000boundedness}, which shows that the stability of a switched linear system is undecidable.
\end{proof}

In practice, the pulses that implement unitary gates in NISQ devices do so in a discrete and finite fashion. Therefore, we reformulate Problem \eqref{p:QAOAstability} to a \emph{digitized version} as follows.

\begin{prob}[Digitized {QAOA} convergence]\label{p:QAOAstabilityDigital}\hfill
	\begin{description}[noitemsep,leftmargin=0.5cm,font=\normalfont]
		\item[Instance] 
		Two Hamiltonians $H_b,H_c\in \Herm(\mathbb{C}^n)$ and the number of layers $L$ in unary notation.
		Parameters $\beta_i, \gamma_i \in \mathbb{D}$ come from some finite, discrete set
		 $\mathbb{D}$ of numbers representable by a single word in a digital computer.
		\item[Decision version] For a tunable state $\ket{\Psi(\vec \beta, \vec \gamma)}\coloneqq U_b(\beta_L)U_c(\gamma_L)\cdots U_b(\beta_1)U_c(\gamma_1) \ket{\Psi_0}$, where $\ket{\Psi_0}$ is the ground state of $H_b$, $U_b(\beta)=\e^{-\i H_b \beta}$ and $U_c(\gamma)=\e^{-\i H_c \gamma}$,
		decide whether there exists a limit to the sequence produced by {QAOA}. 
	\end{description}
\end{prob}

\begin{prop}\label{th:QAOAstabilityDigital}
Digitized {QAOA} convergence is undecidable. 
\end{prop}

\begin{proof}
Since $\beta_i,\gamma_i$ come from some finite, discrete ordered set, by setting $U_{i} = U_b(\beta_i)U_c(\gamma_i)$, the vocabulary $\CU_{(L)} = \{ U_1,\ldots, U_L \}$ is finite. 
The convergence of products of operators of the form $\Pi_U= U_n \ldots U_1$ can be determined by computing the joint spectral radius $\hat{\rho}(\mathcal{U})$. This is an instance of Problem \eqref{p:BMP}. It is known that convergence of $\Pi_U$ holds only if $\hat{\rho}(\mathcal{U}_{(L)}) <1$ which is itself an instance of Problem \eqref{p:USR}. It was shown in \cite{blondel2000boundedness}, that the USR is undecidable implying that BMP is undecidable too. The proof is based on a reduction of the problem of probabilistic finite automata emptiness to USR and it amounts into showing that $\hat{\rho}(\mathcal{U}_{(L)}) >1$ and $\hat{\rho}(\mathcal{U}_{(L)})\leq 1$. The proof also covers QAOA. 
\end{proof}


\begin{prop}
\label{th:QAOAstabilityDigital}
Assuming Conjecture \ref{conj:system} holds, there exists no recursive function to decide whether the digitized {QAOA} convergence problem has a solution.
\end{prop}

\begin{prop}
\label{th:QAOAstabilityDigital}
Digitized QAOA is undecidable for the product of two matrices ($L=1$).
\end{prop}

\begin{proof}
To ease the notation, we skip the argument of the unitaries since it is fixed. The set $\CU_{(1)} = \{ U_b, U_c \}$ defined a dictionary out of which one can form words of length $L$. The decidability problem for $\CU_{(1)}$ can be deduced from the proof of Theorem \eqref{th:QAOAstabilityDigital}: First, one defines the set of two $nL\times nL$ matrices $\tilde{\CU}_{(1)}= \{ U,T\}$ where $U = \mathrm{diag}(U_1,\ldots, U_L )$ and 
\begin{align}
    T = \begin{pmatrix} 0 & \mathds{1}_{n(L-1)} \\ \mathds{1}_n & 0 \end{pmatrix},
\end{align}
where $\mathbb{I}_n$ is the $n\times n$ identity matrix. Then it is easy to verify that $\rho(\CU_{(n)})\leq 1$ if and only if $\rho(\tilde{\CU}_{(1)})\leq 1$ \cite{blondel2000boundedness}. This is based on the proof of Theorem \ref{Th:TwoMatrices}.
\end{proof}


\section{Conclusions}\label{sec:conclusions}

We have shown that training VQA is undecidable, assuming a certain conjecture on the solvability of a certain structured system of equations in complex numbers.

The investigation has shown that the undecidability of VQA via universal Diophantine quations is not completely straightforward. Instead, the undecidability is related to the solvability of an overdetermined, but highly-structured complex-variate system of equations. Algorithmically, confirmation of the conjecture would confirm the broadly shaping consensus that VQA is rather ``heuristic'', 
with behaviour hard to characterize and whose use across a wide enough class of possible instances cannot be fully reliable.
While this may not be an issue in many applications of optimization in engineering and chemistry, it seems unlikely that applications in number theory \cite{yan2022factoring} would be successful.

Note that affirmation of the conjecture does not rule out computability in specific instances, 
or even specific broad classes of Hamiltonians.
Indeed, computability has been shown for large-girth regular graphs \cite{basso2021quantum},
projectors \cite{zhou2020quantum},
and the {S}herrington-{K}irkpatrick model \cite{farhi2019quantum,basso2021quantum}.


\begin{acknowledgments}
This work has been partially supported by project MIS 5154714 of the National Recovery and Resilience Plan Greece 2.0 funded by the European Union under the NextGenerationEU Program.
\end{acknowledgments}

\section*{DISCLAIMER}
This paper was prepared for information purposes and is not a product of HSBC Bank Plc. or its affiliates.
Neither HSBC Bank Plc. nor any of its affiliates make any explicit or implied representation or warranty and none of them accept any liability in connection with this paper, including, but not limited to, the
completeness, accuracy, reliability of information contained herein and the potential legal, compliance,
tax or accounting effects thereof. Copyright HSBC Group 2025.


\bibliographystyle{myapsrev}
\bibliography{refs} 


\appendix

\section{Background and Related Work}
\label{sec:background}

\subsection{Variational Quantum Algorithms}\label{sec:appVQA}

Variational Quantum Algorithms comprise a category of hybrid quantum-classical algorithms \cite{cezero2021}. A classical optimization problem is translated into a physical quantum problem associated with a Hamiltonian $H$. The optimization process involves minimizing the ground state energy corresponding to Hamiltonian $H$. Once this mapping is established, execution on a quantum computer requires preparing an initial state $|\psi_0\rangle$ and allowing it to evolve according to $H$ through a parametrized quantum circuit, with parameter vector $\vec\phi \in \mathbb{R}^L$, also referred to as the variational form. Consequently, the evolved state $|\psi(\vec\phi)\rangle = U(\vec\phi)|\psi_0\rangle$ is obtained.
\begin{figure}[bh!]
\centering
\includegraphics[width=0.40\textwidth]{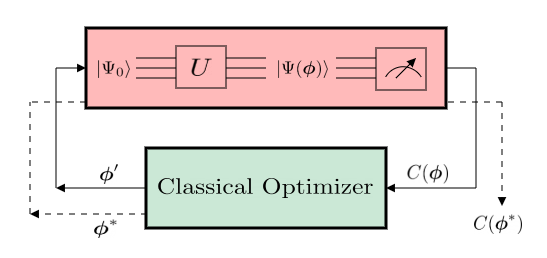}
\caption{Schematic of the architecture of variational quantum algorithms. The classical optimizer undertakes the task to update the parameters $\vec \phi \to \vec \phi'$ such that the cost function is minimized. The loop halts once optimal parameters have been found.}
\label{fig:1}
\end{figure}
The unitary operator $U$, that is being implemented by a set of gates in the parametrized quantum circuit, approximates $ {\rm e}^{-{\rm i}H(\vec \phi)t}$ within some error margin where $t$ denotes time and $H \in \Herm(\mathbb{C}^n)$ the (parametrized) system Hamiltonian. The optimal values of the parameters are not known so a random initialization is due. Subsequently, expectation values of specific observables are measured and a cost function $C(\vec \phi)$ is constructed that is fed into a classical optimizer that minimizes it. Once an optimal value is found, the classical computer updates the values for the parameters $\vec \phi$ into the quantum circuit and the procedure repeats.

QAOA is a variational algorithm whose objective is to approximately solve certain types of combinatorial optimization problems such as \MaxCut~and~\ThreeSAT~problems. QAOA uses the adiabatic theorem which states that a slow enough transition between two Hamiltonians $H_b$ and $H_c$ guarantees that the ground state $|\psi_b\rangle$ of $H_b$ slowly transitions to the ground state $|\psi_c\rangle$ of $H_c$ as long as the Hamiltonians are gapped and level crossings are avoided ~\cite{albash2018adiabatic}.

Typically, the QAOA first creates an initial state $|+\rangle^{\otimes n}$, by applying $n$ Hadamard gates to $|\vec 0\rangle^{\otimes n}$, that evolves according to the mixer Hamiltonian $H_b = \sum_{i=0}^{n-1}{X}_{i}$. The quantum circuit applies successively $L$ layers of the unitaries $\prod_n\e^{-\i \beta_n {H}_b}\e^{-\i \gamma_n {H}_c}$, $n=0,
\ldots, L$ to create a trial state $|\Psi(\vec \beta,\vec \gamma)\rangle$. The parameters $\vec \beta, \vec \gamma$ are then fed to a classical optimizer so as to be updated for the next iteration of the algorithm.

VQAs in general, and QAOA specifically, needs to overcome several difficulties in order to perform reasonably efficient computations. One such difficulty is the challenge to find variational forms such that resulting cost function provides with optimal solutions $\hat{p}_t$, at each iteration $t$, that minimize the error $\varepsilon_p$ to the actual but unknown optimum $p^*$ of the optimization problem:
\begin{align}
    \varepsilon_p = |\hat p - p^*|.
\end{align}
However, this is heavily constrained by the hardware capabilities of NISQ devices including gate fidelities, decoherence times, circuit depth, etc. Once a solution $\hat p_t$ has been found, the classical optimizer needs to perform efficient search within the parameter landscape such that the estimation of the next iteration is closer to $p^*$ and the error converges to zero, $\lim_r \varepsilon_p^{(r)}\to 0$, for finite $r\in \mathbb{Z}_{\geq 1}$ being the iteration number.

\subsection{Complexity of Classical Subproblems in Variational Quantum Algorithms}

There exists a multitude of formalizations for variational quantum algorithms. In this work we use the notation of \cite{bittel2021training}, presented the following: 

\begin{prob}[{VQA} minimization, oracular formulation, paraphrasing \cite{bittel2021training}]\label{p:VQA_O}\hfill
	\begin{description}[noitemsep,leftmargin=0.5cm,font=\normalfont]
		
		\item[Instance] 
		A set of generators $\{H_i \}_{i \in \{1,\dots, L\}}$ 
		and an observable $O$ acting on $\H=(\mathbb{C}^2)^{\otimes N}$, 
		given in terms of their Pauli basis representation. 
		
		\item[Oracle access] 
		We set $\ket{\Psi(\vec \phi)}\coloneqq U_L(\phi_L)\cdots U_1(\phi_1)\ket{\vec 0}$ with
		\begin{align}\label{Eq:1}
		    U_i(\phi_i)=\e^{-\i H_i \phi_i}.
		\end{align}
		 The oracle $\mc O$ returns
		$\OEx{\vec \phi} \coloneqq \sandwich{\Psi(\vec \phi)}{O}{\Psi(\vec \phi)}$, given $\vec \phi$, up to any desired polynomial additive error. 
		
		\item[Optimization version]
		Find $\vec \phi \in \mathbb{R}^L$ that minimizes $\OEx{\vec \phi}$ provided access to $\mc O$. 

		\item[Decision version for $a \in \mathbb{R}$]
		Decide whether there exists $\vec \phi \in \mathbb{R}^L$ whose value $\OEx{\vec \phi} \le a$ provided access to $\mc O$. 
	\end{description}
\end{prob}

\begin{theorem}[Hardness of {VQA} optimization, decision version with the oracular formulation, paraphrasing \cite{bittel2021training}]
	\label{thm:VQA1}
	Assuming $\P\neq \NP$ there is no deterministic classical algorithm that solves the decision version of Problem~\ref{p:VQA_O} in polynomial time.
\end{theorem}

Alternatively, with computational basis of the Hilbert space $\H$ of dimension $\dim(\H)=2^N\eqqcolon n$ we formulate the following problem:

\begin{prob}[{VQA} minimization problem, paraphrasing \cite{bittel2021training} ]\label{p:VQA:appendix}\hfill
\begin{description}[noitemsep,leftmargin=0.5cm,font=\normalfont]
\item [Instance] 
An initial state $\ket{\Psi_0}\in \mathbb{C}^n$, a set of generators $\{H_i \}_{i \in \{1,\dots, L\}}\in \Herm(\mathbb{C}^n)$, where $L$ is the number of layers and an observable $O\in\Herm(\mathbb{C}^n)$. 

\item[Optimization version]
For $\ket{\Psi(\vec \phi)}\coloneqq U_L(\phi_L)\cdots U_1(\phi_1)\ket{\Psi_0}$ with $U_i(\phi_i)=\e^{-\i H_i \phi_i}$, 
find $\vec \phi \in \mathbb{R}^L$ that minimizes $\OEx{\vec\phi}\coloneqq~\sandwich{\Psi(\vec \phi)}{O}{\Psi(\vec \phi)}$. 

\item[Decision version for $a \in \mathbb{R}$]
For $\ket{\Psi({\vec \phi})}\coloneqq U_L(\phi_L)\cdots U_1(\phi_1)\ket{\Psi_0}$ with $U_i(\phi_i)~=~\e^{-\i H_i \phi_i}$, 
	determine if there exists 
	$\vec \phi \in \mathbb{R}^L$ for which $\sandwich{\Psi(\vec \phi)}{O}{\Psi(\vec \phi)} \leq a$. 
		\end{description}
\end{prob}

\begin{theorem}[Hardness of {VQA} optimization, paraphrasing \cite{bittel2021training}]\label{thm:VQA2}
	Optimization version of {VQA} minimization problem (Problem~\ref{p:VQA}) is \NP-hard. 
	Decision version of {VQA} minimization problem (Problem~\ref{p:VQA}) is \NP-hard for $L=1$ layer.
\end{theorem}\label{sec:VQAopt_QAOA}
This result implies that even when simplifying the problem to a single layer, the decision version remains computationally difficult. 
VQA minimization involves finding a global minimum of a multimodal objective landscape, a \NP-hard problem. The similarity lies in that their objective functions exhibit non-linear and non-convex properties, leading to numerous local minima and complicating the search for global optima. Famously, Blum and Rivest \cite{Blum1992} showed that training a 3-neuron multi-layer perceptron, which has a similar landscape, is \NP-Complete. (See \cite{ma2002} for relevant results and \cite{Slzer2021} for some recent work on the topic from the neural network reachability point of view.)

Next we proceed to discuss analogous results about QAOA, the algorithm we will dominantly use for the main result of this paper in Sec. \ref{MainResult}.

\begin{prob}[{QAOA} minimization problem, \cite{bittel2021training}]\label{p:QAOA:app}\hfill
	\begin{description}[noitemsep,leftmargin=0.5cm,font=\normalfont]
		\item[Instance] Two Hamiltonians $H_b,H_c\in \Herm(\mathbb{C}^n)$ and the number of layers $L$ in unary notation\footnote{This means that the length of the input scales linearly with $L$.}.
		
		\item[Optimization version] For a tunable state $\ket{\Psi(\vec \beta, \vec \gamma)}\coloneqq U_b(\beta_L)U_c(\gamma_L)\cdots U_b(\beta_1)U_c(\gamma_1) \ket{\Psi_0}$, where $\ket{\Psi_0}$ is the ground state of $H_b$, $U_b(\beta)=\e^{-\i H_b \beta}$ and $U_c(\gamma)=\e^{-\i H_c \gamma}$,
		find $\vec \beta,\vec \gamma \in \mathbb{R}^d$ which minimize $\OEx{\vec \beta, \vec \gamma}\coloneqq\sandwich{\Psi(\vec \beta, \vec \gamma)}{H_c}{\Psi(\vec \beta, \vec \gamma)}$. 
		
\item[Decision version for $a \in \mathbb{R}$] 
For a tunable state $\ket{\Psi(\vec \beta, \vec \gamma)}\coloneqq U_b(\beta_L)U_c(\gamma_L)\cdots U_b(\beta_1)U_c(\gamma_1) \ket{\Psi_0}$, where $\ket{\Psi_0}$ is the ground state of $H_b$, $U_b(\beta)=\e^{-\i H_b \beta}$ and $U_c(\gamma)=\e^{-\i H_c \gamma}$,
		decide if there exists  $\vec \beta,\vec \gamma \in \mathbb{R}^d$ such that $\OEx{\vec \beta, \vec \gamma}\coloneqq\sandwich{\Psi(\vec \beta, \vec \gamma)}{H_c}{\Psi(\vec \beta, \vec \gamma)} \le a$.
		
	\end{description}
\end{prob}

\begin{theorem}[Hardness of optimization in {QAOA}, \cite{bittel2021training}]\label{thm:QAOA}
	Decision version of {QAOA} minimization problem (Problem~\ref{p:QAOA}) is \NP-hard for $L=1$ layer.
\end{theorem}
\begin{proof}
	To show this, we show a reduction from single layer {VQA} to {QAOA}, which implies that Problem~\ref{p:QAOA} is \NP-hard.
\end{proof}

The proof of Theorem~\ref{thm:QAOA} is based on the proof that VQA is \NP-hard even for $L=1$. In the subsequent section, we show that \NP-hardness is a very conservative estimate. In particular, we show that certain problems associated to VQAs and QAOAs are  undecidable. 

\subsection{Undecidability in Quantum Computing}

Undecidability is a much stronger notion than \NP-hardness.  
Instead of proving a conditional statement on the non-existence of a polynomial-time algorithm for obtaining a solution, e.g. Proposition \ref{thm:VQA1},
undecidability is an unconditional statement on the non-existence of a finite-time algorithm for indicating whether a solution satisfying a certain objective even exists.

While there is a long history of work on undecidability in physics \cite{svozil1993randomness},
starting with the influential paper \cite{da1991undecidability}
on the undecidability in classical mechanics,
only a handful of undecidability results are known in quantum information theory and quantum computing. 
Following the pioneering work of Lloyd \cite{lloyd1993quantum,lloyd2017uncomputability},  Cubitt et al. \cite{cubitt2015undecidability} have shown the undecidability of the spectral gap, even in one dimension \cite{bausch2020undecidability}, and uncomputability \cite{bausch2021uncomputability} of phase diagrams.
Smith \cite{smith2006three} has shown that 
no fully general form of the quantum adiabatic theorem can exist that yields computable upper bounds on adiabatic convergence times.
Bondar and Pechen \cite{bondar2020uncomputability} have shown uncomputability of a digitized quantum optimal control.

\subsection{Undecidability in Control Theory}

There is, however, a much longer tradition of the study of undecidability and uncomputability in control theory \cite{blondel2000survey}.
Using the notion of the generalized spectral radius \cite{daubechies1992sets}, \cite{blondel2000boundedness} show that the convergence of products of operators of the form of $U_L(\phi_L)\ldots U_1(\phi_1)$ is undecidable. 

The joint spectral radius \cite{rota1960note} of a set of matrices is a measure of the maximal asymptotic growth rate that can be obtained by forming long products of matrices taken
from a set $\CA = \{A_1,\ldots,A_n \}$ of real $n\times n$ matrices.

Consider products $A_{(t)} = A_{t}\ldots A_1$. For example, if $\CA = \{ A_1,A_2\}$ and $t=2$,
\begin{align}
    A_{(2)} = 
    \begin{cases}
        A_1^2 \\ A_2^2 \\ A_1A_2 \\ A_2A_1 \
    \end{cases}.
\end{align}
The spectral radius of a square matrix $A \in \CA$ is defined as:
\begin{align}
    \rho(A) = \mathrm{max}\{|\lambda| \text{ such that } \lambda \text{ is an eigenvalue of } A  \}.
\end{align}
A natural generalization of the spectral radius that serves as a measure of growth ($t\to \infty$) of matrix products $A_{(t)}$ is the \emph{joint spectral radius}:
\begin{align}
    \hat{\rho}(\CA) &= \lim_{t\to \infty} \sup \hat{\rho}_t(\CA),
\end{align}
where for some matrix norm $\lVert \cdot \rVert$ we define
\begin{align}
    \hat{\rho}_t(\CA) = \max_{A_i\in \CA}\lVert A_t\ldots A_1 \rVert^{1/t}.
\end{align}

A very similar concept is that of the
\emph{generalized spectral radius} $\rho(\CA)$ is defined as:
\begin{align}
    \rho(\CA) = \lim_{t\to \infty} \sup \rho_t(A_t\ldots A_1),
\end{align}
where
\begin{align}
    \rho_t(\CA) = \max_{A_i \in \CA} \rho(A_t\ldots A_1)^{1/t}.
\end{align}
It is known that $\rho_t(\CA) \leq \rho(\CA)$ \cite{Wang2013}, while for any finite set of matrices $\CA$ it holds that \cite{BERGER199221}
\begin{align}
    \hat{\rho}(\CA) = \rho(\CA).
\end{align}

\begin{prob}[Unit Spectral Radius ({USR})]\label{p:USR}\hfill
  \begin{description}
     \item[Input] A finite set $\CA$ of $n\times n$ $\mathbb{Q}$-valued matrices.
     \item[Decision version for $a = 1$] Is the joint spectral radius $\rho(\CA)\leq a$?
\end{description}
\end{prob}

\begin{prob}[Bounded Matrix Products (BMP)]\label{p:BMP}\hfill
  \begin{description}
     \item[Input] A finite set $\CA$ of $n\times n$ $\mathbb{Q}$-valued matrices.
     \item[Decision version]  Is the set of all matrix products bounded?
\end{description}
\end{prob}

Attacking Problem \eqref{p:USR} and Problem \eqref{p:BMP} amounts to the asymptotic computation of the joint spectral radius for $\CA$. If it is found to be less than one then the problem is decidable and the corresponding (switched) linear system is stable \cite{blondel2000boundedness}.


\begin{theorem}[Blondel and Tsitsiklis,  \cite{blondel2000boundedness}]\label{Th:TwoMatrices}
The problems BMPs
and USR are undecidable. They
remain undecidable even in the special case where $\CA$
consists of only two matrices.
\end{theorem}


\subsection{Undecidability in Optimization}

Optimization without any assumptions on the objective function is undecidable.  
This builds upon the work of Mati{\^a}sevi{\v{c}} from the 1970 \cite{matiyasevich1970Diophantineness,englishproof,Matiasevic1993hilbert}, which has resolved Hilbert's tenth problem \cite{hilbert1902mathematical}. Hilbert's tenth problem \cite{hilbert1902mathematical} asked for a general algorithm which, for any polynomial equation with integer coefficients and a finite number of unknowns that can decide whether the equation has a solution with all unknowns taking integer values. That is, given $p \in \mathbb{Z}[x]$ with $x \in \mathbb{R}^d$, decide if $p(x)=0$ has integer solutions.
Since Hilbert's problem is equivalent to the decision problem of whether the real-valued optimization problem,
\begin{align}
\label{polyplussin}
\min_{x\in\mathbb{R}^d} \, p(x)+\sum\limits_{i=1}^d \sin^2(\pi x_i)
\end{align}
has a zero-valued optimum solution, 
where $p(x)$ is the original polynomial \footnote{
Cf. \cite[Definition 3.10]{da1991undecidability} or  \cite[Proof of Theorem 3]{zhu2006unsolvability} or a modern survey in \cite{liberti2019undecidability}.}, this can be extended to optimization without the integrality of the variables:

\begin{theorem}[Unconstrained continuous global minimization is undecidable, paraphrasing \cite{zhu2006unsolvability}]\label{th:optund}
Given a continuously differentiable function $f(x_1, \ldots ,x_L)$, unconstrained continuous global minimization seeks to find a global minimal solution of $min f$. The associated decision problem is:
Given $L, f$, a constant $B$, does there exist $x_i, i = 1, 2...,L$ such
that $f(x_1, \ldots ,x_L) \le B$?
There exists no recursive function to decide this decision problem.
\end{theorem}

We refer to \cite{liberti2019undecidability} for a nice survey, which fixes some issues with the proofs of \cite{zhu2006unsolvability}.

\section{A Table of Notation}

\begin{table}[h!]
\begin{tabular}{llll} \toprule[2pt]
Variable      & Description            & Space                 &Equation(s) \\ \hline
$U_i(\phi_i)$ & parametrized unitaries & $U(n)$                &\eqref{Eq:1},\eqref{eq:classicalopt11}\\
$O$           & VQA operator           & $\Herm(\mathbb{C}^n)$ &\eqref{eq:O},\eqref{eq:classicalopt11}\\
$\vec{\phi}$  & variational parameters & $\mathbb{R}^{58}$     &\eqref{Eq:1},\eqref{eq:classicalopt11} \\
$\ket{\Psi_0}$& initial state          & $\mathcal{H}$         &\eqref{eq:classicalopt11} \\
$q_j(\phi)$   & summand of SOS polynomial & $\mathbb{R}[\vec \phi]$& \eqref{eq:classicalopt10}  \\
$\varkappa_j$ & Sylvester coefficients    &$\mathbb{C}$ & \eqref{eq:matrixexp}, \eqref{eq:cb} \\
$c_j$         & monomial choosing vector  &$\Herm(\mathbb{C}^n)$ & \eqref{eq:cb} \\
$b_j(\vec{\phi})$& monomial vector        &$\mathbb{R}[\vec{\phi}]$& \eqref{eq:cb} \\
$g_i$            & complex polynomial     & $\mathbb{C}[\vec x, \vec w]$ & \eqref{eq:newsystem} \\ \toprule[1.2pt]
System          & Parameters     & Variables  &Equation \\ \hline 
UDE             & $u, x,y,z$       &  $a,b,c,\ldots, \tau,\varphi$     & \eqref{eq:UDE}  \\ \toprule[1.2pt]
System          & \# Variables     & \# DOF &Equation \\ \hline
Initial   & & &\eqref{eq:classicalopt11} \\
New      & & & \eqref{eq:newsystem} \\ \hline 
\end{tabular}
\caption{An overview of the notation.}
\label{Table:Variables}
\end{table}

\section{An Example of a UDE}

\begin{figure}[h!]
\begin{equation}
\label{eq:UDE}
    \begin{aligned}
(elg^2+ \alpha - (b-xy)q^2)^2 + \\  
(q - b^{5^{60}})^2 + \\ 
(\lambda +q^4 - 1 - \lambda b^5)^2 + \\ 
(\theta +2z - b^5)^2 + \\ 
(u+t\theta - l)^2 + \\ 
(y+m \theta - e)^2 + \\ 
(q^{16} - n)^2 + \\ 
\left([g+eq^3+lq^5+(2(e-z\lambda)(1 +xb^5 +g)^4 + \quad \right. \\
\lambda b^5 + \lambda b^5q^4)q^4] [n^2-n] + \quad \\
\left. [q^3-bl + 1+ \theta \lambda q^3 + (b^5 - 2)q^5] [n^2 - 1] - r \right)^2 + \\
    (2 w s^2 r^2 n^2 - p)^2 + \\ 
(p^2 k^2 - k^2 + 1 - \tau^2)^2  + \\ 
(4(c-k s n^2)^2 + \nu - k^2)^2 + \\ 
(r + 1 +hp-h - k)^2 + \\ 
((wn^2 + 1)rsn^2 - a)^2  + \\ 
(2r+ 1 + \varphi - c)^2 + \\ 
(bw+ca-2c+4a\gamma -5\gamma -d)^2 + \\ 
((a^2 - 1)c^2+1 - d^2)^2 + \\ 
((a^2- 1)i^2c^4+1 - f^2)^2 + \\ 
(((a +f^2(d^2 - a))^2 - 1)(2r + 1 +jc)^2 + 1 - (d + of)^2)^2
\end{aligned}
\end{equation}
\caption{Universal Diophantine equation $D(28,2 \cdot 5^{60})$ based on the work of \cite{Jones1982JONUDE, Matiyasevich1993}, where $a, b, c, \ldots, s, t, w, \alpha, \gamma, \nu, \theta, \lambda, \tau, \varphi$
are variables and $u, x, y, z$ are parameters that encode the Turing machine. Notice that the 
Universal Diophantine equation 
is a sum-of-squares polynomial, which is compatible with \eqref{eq:classicalopt1}.}
\label{fig:UDE}
\end{figure}

\section{Further clarifications and a simple example}
\label{sec:example2}

In order to assist in understanding the form of the equation system, let us consider how a potential two-layer example, instead of the 58-layer case, which is our target, corresponds to an encoding of a degree-8 polynomial. We shall see the precise form of the vectors $\vec c$ and $\vec b$. With $H \in \Herm(\mathbb{C}^5)$, such that the expansion runs for $0\leq j \leq 4$, and utilizing Sylvester's formula, we obtain the following:
\begin{equation}\label{example}
   \begin{aligned}
    c_{(0,0)} & = \varkappa_0\mathds{1} \\
    c_{(1,0)} & = \varkappa_1 H_1 \\
    c_{(0,1)} & = \varkappa_1 H_2 \\
    c_{(2,0)} & = \varkappa_2 H_1^2 \\ 
    c_{(0,2)} &= \varkappa_2 H_2^2  \\
    c_{(1,1)} &= 2\varkappa_2 H_1H_2 \\
    c_{(3,0)} &= \varkappa_3 H_1^3 \\
    c_{(0,3)} &= \varkappa_3 H_2^3 \\
    c_{(2,1)} &= 3\varkappa_3 H_1^2H_2 \\
    c_{(1,2)} &= 3\varkappa_3 H_1H_2^2 \\
    c_{(4,0)} &= \varkappa_4 H_1^4 \\
    c_{(0,4)} & = \varkappa_4 H_2^4 \\
    c_{(3,1)} &= 4\varkappa_4 H_1^3H_2 \\
    c_{(1,3)} &= 4\varkappa_4 H_1H_2^3 \\
    c_{(2,2)} &= 6\varkappa_4 H_1^2H_2^2, \
\end{aligned} 
\end{equation}
where, for simplicity, we absorbed the complex unit into $\{H_i\}$. All dependencies on $\vec \phi$ are encoded in $\vec{b}(\vec{\phi})$. Furthermore, observe that for each monomial degree $i=j+k$, for $c_{(j,k)}$, only the coefficient $\varkappa_i$ is relevant. From this example, it becomes clear that the indexing of $\vec c$ directly reveals the form of the corresponding monomial on the generators $\{H_i\}_{i=1}^{58}$. Furthermore, all $c(j,k)$ contain equal degrees of freedom since they depend, thorough $\varkappa_{j+k}$, on all of the variational parameters,
\begin{align}
    c_{(j,k)} = c_{(j,k)}(\{\phi_i\}, \{H_{58}\}).
\end{align}

\section{An SAGE Script}

\begin{figure}[h!]
\centering
 \begin{verbatim}
var('x,y,a1,a2,a4,b1,b2,c1,c2,d1,d2,a,b')
g1(x,y)=a1*x^2  + a2*y^2 - a 
g2(x,y)=b1*x^2  + b2*y^2 - b
g3(x,y)=c1*x^2  + c2*y^2
g4(x,y)=d1*x^2  + d2*y^2
R.<x,y,a1,a2,b1,b2,c1,c2,d1,d2> = 
   CC[x,y,a1,a2,b1,b2,c1,c2,c3,d1,d2]
I = Ideal(g1, g2, g3, g4)
S = R.quotient_Rig(I)
X = SpecS
S.is_empty()
\end{verbatim}
\caption{A SAGE script for testing the non-emptiness of the moduli space of quadrics over $\mathds{CP}^4$.}
\label{fig:sage}
\end{figure}

\end{document}